%% file: main.tex
\documentclass[envcountsame]{llncs}

\usepackage{amssymb}
\usepackage{amsmath}
\usepackage{pgf}
\usepackage{tikz}
\usepackage{fancyhdr}
\usetikzlibrary{automata}
\usetikzlibrary{shapes,snakes}
\usetikzlibrary{positioning}
\usetikzlibrary{arrows,automata}
\usetikzlibrary{positioning}
\usetikzlibrary{arrows,shapes,fit,backgrounds}

\newcommand{\sep}{\ \Big|\ }

\newcommand{\X}{{\ensuremath{\mathbf{X}}}}
\newcommand{\F}{{\ensuremath{\mathbf{F}}}}
\newcommand{\G}{{\ensuremath{\mathbf{G}}}}
\newcommand{\U}{{\ensuremath{\mathbf{U}}}}
\newcommand{\true}{{\ensuremath{\mathbf{tt}}}}
\newcommand{\false}{{\ensuremath{\mathbf{ff}}}}
\newcommand{\tran}[1]{\stackrel{#1}{\longrightarrow}}
\DeclareMathOperator{\Inf}{Inf}
\newcommand{\I}{\mathcal I}
\newcommand{\nxt}{{\ensuremath{\mathrm{next}}}}
\newcommand{\B}{B}
\newcommand{\elemB}{\mathit{wf}}
\newcommand{\product}{\overline{\mathcal M}}
\newcommand{\pred}{\heartsuit}
\newcommand{\lift}{\mathrm{Lift}}

\renewcommand{\L}{L}
\newcommand{\D}{R}
\newcommand{\GR}{GR}

\tikzset{
    state/.style={
		rectangle,
            rounded corners,
            draw=black,
            minimum height=2em,
            minimum width=2em,
            inner sep=4pt,
            text centered,
            }
}

\newcommand{\myspace}{\vspace{-1em}}
\newcommand{\myspaceb}{\vspace{-.5em}}
\newcommand{\myspacec}{\vspace{-.3em}}

\title{Automata with Generalized Rabin Pairs for 
\\
Probabilistic Model Checking and LTL Synthesis\myspaceb}

\author{\myspace Krishnendu Chatterjee\inst{1}\thanks{The author is supported by Austrian Science Fund (FWF) Grant No P 23499-N23,  FWF NFN Grant No S11407-N23 (RiSE), ERC Start grant (279307: Graph Games), and Microsoft faculty fellows award.} \and
Andreas Gaiser\inst{2}\thanks{The author is supported by the DFG Graduiertenkolleg 1480 (PUMA).} \and
Jan K\v ret\'insk\'y\inst{2,3}\thanks{The author is supported by the Czech Science Foundation, project No. P202/10/1469}}

\institute{
IST Austria \and
Fakult\"at f\"ur Informatik, Technische Universit\"at M\"unchen, Germany \and 
Faculty of Informatics, Masaryk University, Brno, Czech Republic}

\begin{document}
\pagestyle{plain}

\maketitle

 \renewcommand{\thefootnote}{\fnsymbol{footnote}}


\input{intro}


\myspace\myspace\myspaceb

\section{Preliminaries}

\myspace

In this section, we recall the notion of linear temporal logic (LTL) and illustrate the recent translation of its (\F,\G)-fragment to DRW~\cite{cav,atva} through the intermediate formalism of DGRW. Finally, we define an index that is important for characterizing the savings the new formalism of DGRW brings as shown in the subsequent sections.

\myspace

\subsection{Linear temporal logic}

We start by recalling the fragment of linear temporal logic with \emph{future} (\F) and \emph{globally} (\G) modalities.

\begin{definition}[LTL(\F,\G) syntax]\label{def:ltl-syn}
The formulae of the \emph{(\F,\G)-fragment of linear temporal logic}
are given by the following syntax:
\myspaceb

$$\varphi::= a\mid \neg a\mid \varphi\wedge\varphi \mid \varphi\vee\varphi \mid \F\varphi \mid \G\varphi$$
\myspaceb

 \noindent where $a$ ranges over a finite fixed set $Ap$ of atomic propositions.
\end{definition}
We use the standard abbreviations $\true:=a\vee\neg a$ and $\false:=a\wedge\neg a$.
Note that we use the negation normal form, as negations can be pushed inside to atomic propositions due to the equivalence of $\F\varphi$ and $\neg\G\neg\varphi$.

\begin{definition}[LTL(\F,\G) semantics]
Let $w\in (2^{Ap})^\omega$ be a word. The $i$th letter of $w$ is denoted $w[i]$, i.e.~$w=w[0]w[1]\cdots$. 
Further, we define the $i$th suffix of $w$ as $w_i =w[i]w[i+1]\cdots$. The semantics of a formula on $w$ is then defined inductively as follows: $w \models a \iff a \in w[0]$; $w \models \neg a \iff a \notin w[0]$; $w\models \varphi \wedge \psi \iff w \models \varphi \text{ and }  w \models \psi$; $w \models \varphi \vee \psi \iff w \models \varphi \text{ or }  w \models \psi$; and
\myspace

{\allowdisplaybreaks
\begin{align*}
  w &\models \F \varphi &&\iff \exists \, k\in\mathbb N_0: w_k \models
  \varphi\\
  w &\models \G \varphi &&\iff \forall \, k\in\mathbb N_0: w_k \models
  \varphi
 \end{align*}}
\end{definition}
\myspace\myspace

\subsection{Translating LTL(\F,\G) into deterministic $\omega$-automata}

Recently, in~\cite{cav,atva}, a new translation of LTL(\F,\G) to deterministic automata has been proposed. This construction avoids Safra's determinization and makes direct use of the structure of the formula. We illustrate the construction in the following examples.

\begin{example}\label{ex:trans-fin}
Consider a formula $\F a\,\vee\,\G b$. The construction results in the following automaton. The state space of the automaton has two components. The first component stores the current formula to be satisfied. Whenever a letter is read, the formula is updated accordingly. For example, when reading a letter with no $b$, the option to satisfy the formula due to satisfaction of $\G b$ is lost and is thus reflected in changing the current formula to $\F a$ only. 

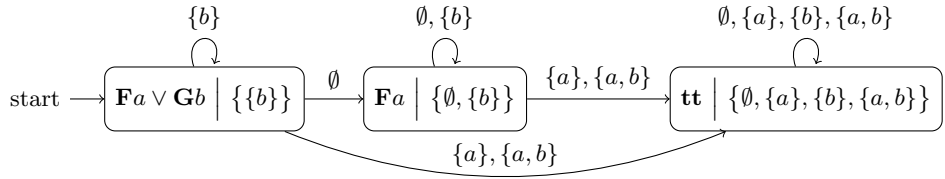
\hspace*{-2em}
\begin{tikzpicture}[x=4cm,y=2cm,font=\footnotesize,yscale=0.7]
\node[state,initial] (i) at (0,0) {$\F a\vee\G b \sep \big\{\{b\}\big\}$};
\node[state] (ii) at (0.8,0) {$\F a \sep \big\{\emptyset,\{b\}\big\}$};
\node[state] (iii) at (2,0) {$\true \sep \big\{\emptyset,\{a\},\{b\},\{a,b\}\big\}$};

\path[->] (i) edge[loop above] node[above]{$\{b\}$} (i)
          (i) edge node[above]{$\emptyset$} (ii)
          (i) edge[bend right] node[above]{$\{a\},\{a,b\}$} (iii)
          (ii) edge[loop above] node[above]{$\emptyset,\{b\}$} (ii)
          (ii) edge node[above]{$\{a\},\{a,b\}$} (iii)
          (iii) edge[loop above] node[above]{$\emptyset,\{a\},\{b\},\{a,b\}$} (iii);
\end{tikzpicture}

The second component stores the last letter read (actually, an equivalence class thereof). The purpose of this component is explained in the next example. For formulae with no mutual nesting of $\F$ and $\G$ this component is redundant.

The formula $\F a\,\vee\,\G b$ is satisfied either due to $\F a$ or $\G b$. Therefore, when viewed as a Rabin automaton, there are two Rabin pairs. One forcing infinitely many visits of the third state ($a$ in $\F a$ must be eventually satisfied) and the other prohibiting infinitely many visits of the second and third states ($b$ in $\G b$ must never be violated). The acceptance condition is a disjunction of these pairs.
\end{example}

\myspaceb

\begin{example}\label{ex:trans-infin}
Consider now the formula $\varphi=\G\F a\wedge \G\F\neg a$. Satisfaction of this formula does not depend on any finite prefix of the word and reading $\{a\}$ or $\emptyset$ does not change the first component of the state. This infinitary behaviour requires the state space to record which letters have been seen infinitely often and the acceptance condition to deal with that. In this case, satisfaction requires visiting the second 
state infinitely often \emph{and} visiting the first state infinitely often.\myspaceb

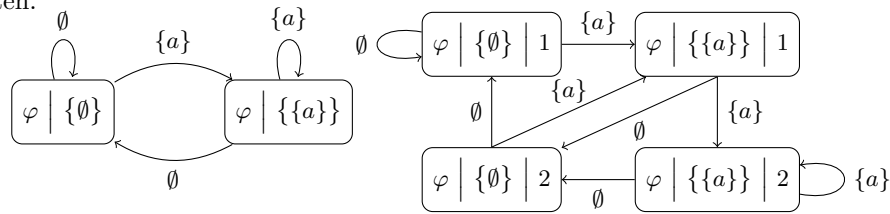
\begin{tikzpicture}[x=3cm,y=1.8cm,font=\footnotesize]
\node[state] (o1) at (0.1,-0.5) {$\varphi\sep\big\{\emptyset\big\}$};
\node[state] (a1) at (1.1,-0.5) {$\varphi\sep \big\{\{a\}\big\}$};
\node (x) at (1.3,-1.7) {};

\path[->] 
(o1) edge[loop above] node[above]{$\emptyset$} (o1)
    edge[bend left] node[above]{$\{a\}$} (a1)
(a1) edge[loop above] node[above]{$\{a\}$} (a1)
    edge[bend left] node[below]{$\emptyset$} (o1)
;
\node[state] (o) at (2,0) {$\varphi\sep\big\{\emptyset\big\}\sep 1$};
\node[state] (a) at (3,0) {$\varphi\sep \big\{\{a\}\big\}\sep 1$};
\node[state] (o2) at (2,-1) {$\varphi\sep\big\{\emptyset\big\}\sep 2$};
\node[state] (a2) at (3,-1) {$\varphi\sep \big\{\{a\}\big\}\sep 2$};

\path[->] 
(o) edge[loop left,max distance=8mm,in=190,out=170,looseness=10] node[left]{$\emptyset$} (o)
    edge node[above]{$\{a\}$} (a)
(a) edge node[right]{$\{a\}$} (a2)
(a.south) edge node[below]{$\emptyset$} (o2)
;
\path[->] 
(o2) edge node[left]{$\emptyset$} (o)
(o2.north) edge node[above]{$\{a\}$} (a)
(a2) edge[loop right, max distance=8mm,in=10,out=-10,looseness=10] node[right]{$\{a\}$} (a2)
    edge node[below]{$\emptyset$} (o2)
;
\end{tikzpicture}
\myspace\myspace\myspace\myspace

However, such a conjunction cannot be written as a Rabin condition. In order to get a Rabin automaton, we would duplicate the state space. In the first copy, we wait for reading $\{a\}$. Once this happens we move to the second copy, where we wait for reading $\emptyset$. Once we succeed we move back to the first copy and start again. This bigger automaton now allows for a Rabin condition. Indeed, it is sufficient to infinitely often visit the ``successful'' state of the last copy as this forces infinite visits of ``successful'' states of all copies.
\end{example}

In order to obtain a DRW from an LTL formula, \cite{cav,atva} first constructs an automaton similar to DGRW (like the one on the left) and then the state space is blown-up and a DRW (like the one on the right) is obtained. However, we shall argue that this blow-up is unnecessary for application in probabilistic model checking and in synthesis. This will result in much more efficient algorithms for complex formulae. In order to avoid the blow-up we define and use DGRW, an automaton with more complex acceptance condition, yet as we show algorithmically easy to work with and efficient as opposed to e.g.\ the general Muller condition.

\myspace

\subsection{Automata with generalized Rabin pairs}\label{ssec:dgrw}

\myspacec

In the previous example, the cause of the blow-up was the conjunction of Rabin conditions. In~\cite{cav}, a generalized version of Rabin condition is defined that allows for capturing conjunction. It is defined as a positive Boolean combination of Rabin pairs. Whether a set $\Inf(\rho)$ of states visited infinitely often on a run $\rho$ is accepting or not is then defined inductively as follows:

\myspace
\begin{align*}
  \Inf(\rho) &\models \varphi \wedge \psi &&\iff \Inf(\rho) \models \varphi \text{ and }
  \Inf(\rho) \models \psi\\
  \Inf(\rho) &\models \varphi \vee \psi &&\iff \Inf(\rho) \models \varphi \text{ or }
  \Inf(\rho) \models \psi\\
  \Inf(\rho) &\models(F,I) &&\iff F\cap \Inf(\rho)=\emptyset \text{ and } I\cap \Inf(\rho)\neq\emptyset 
\end{align*}
Denoting $Q$ as the set of all states, $(F,I)$ is then equivalent to $(F,Q)\wedge(\emptyset,I)$. Further, $(F_1,Q)\wedge(F_2,Q)$ is equivalent to $(F_1\cup F_2,Q)$. Therefore, one can transform any such condition into a disjunctive normal form 
and obtain a condition of the following form:\myspace
\begin{align*}
\bigvee_{i=1}^k \left( \Big(F_i,Q\Big) \wedge\bigwedge_{j=1}^{\ell_i}\Big( \emptyset, I_i^j \Big) \right) \tag{$*$}
\end{align*}\myspace

Therefore, in this paper we define the following new class of $\omega$-automata:

\begin{definition}[DGRW]
An \emph{automaton with generalized Rabin pairs (DGRW)} is a (deterministic) $\omega$-automaton $\mathcal A=(Q,q_0,\delta)$ over an alphabet $\Sigma$, where $Q$ is a set of states, $q_0$ is the initial state, $\delta:Q\times\Sigma\to Q$ is a transition function, together with a \emph{generalized Rabin pairs  (GRP) acceptance condition} $\mathcal{GR}\subseteq 2^{2^Q\times 2^{2^Q}}$. A run $\rho$ of $\mathcal A$ is accepting for $\mathcal{GR}=\big\{\big(F_i, \{I_i^1, \ldots, I_i^{\ell_i}\}\big) \sep i\in\{1,\ldots,k\}\big\}$ if there is $i\in\{1,\ldots,k\}$ such that
\begin{align*}
F_i\cap \Inf(\rho)&=\emptyset \text{ and }\\
I_i^j\cap \Inf(\rho)&\neq\emptyset \text{ for every } j\in\{1,\ldots,\ell_i\}
\end{align*}
\myspace

Each $(F_i,\I_i)= \Big(F_i, \{I_i^1, \ldots, I_i^{\ell_i} \}\Big)$ is called a \emph{generalized Rabin pair (GRP)}, and the \emph{GRP condition} is thus a disjunction of generalized Rabin pairs..
\end{definition}
W.l.o.g.~we assume $k>0$ and $\ell_i>0$ for each $i\in\{1,\ldots,k\}$ (whenever $\ell_i=0$ we could set $\I_i=\{Q\}$). Although the type of the condition allows for huge instances of the condition, the construction of~\cite{cav} (producing this disjunctive normal form) guarantees efficiency not worse than that of the traditional determinization approach. For a formula of size $n$, it is guaranteed that $k\leq 2^n$ and $\ell_i\leq n$ for each $i\in\{1,\ldots,k\}$. Further, the size of the state space is at most $2^{\mathcal O(2^n)}$. Moreover, consider \emph{``infinitary''} formulae, where each atomic proposition has both $\F$ and $\G$ as ancestors in the syntactic tree of the formula. Since the first component of the state space is always the same, the size of the state space is bounded by $2^{|Ap|}$ as the automaton only remembers the last letter read. We will make use of this fact later.

\subsection{Degeneralization}

As already discussed, one can blow up any automaton with generalized Rabin pairs and obtain a Rabin automaton. We need the following notation. For any $n\in\mathbb N$, let $[1..n]$ denote the set $\{1,\ldots,n\}$ equipped with the operation $\oplus$ of cyclic addition, i.e.~$m\oplus1=m+1$ for $m<n$ and $n\oplus1=1$. 

The DGRW defined above can now be degeneralized as follows.
For each $i\in\{1,\ldots,k\}$, multiply the state space by $[1..\ell_i]$ 
to keep track for which $I_i^j$ 
we are currently waiting for. Further, adjust the transition function so that we leave the $j$th copy once we visit $I_i^j$ and immediately go to the next copy. Formally, for $\sigma\in\Sigma$ 
set $(q,w_1,\ldots,w_k)\tran{\sigma}(r,w_1',\ldots,w_k')$ if $q\tran{\sigma}r$ and $w_i'=w_i$ for all $i$ with $q\notin I_i^{w_i}$ and $w_i'=w_i\oplus1$ otherwise.

The resulting blow-up factor is then the following:

\begin{definition}[Degeneralization index]
For a GRP condition $\mathcal{GR}=\{(F_i,\I_i)\mid i\in[1..k]\}$, we define the \emph{degeneralization domain} 
$\B:=\prod_{i=1}^k [1..|\I_i|]$ and the \emph{degeneralization index of $\mathcal{GR}$} to be $|\B|=\prod_{i=1}^k |\I_i|$.
\end{definition}
The state space of the resulting Rabin automaton is thus $|\B|$-times bigger and the number of pairs stays the same. Indeed, for each $i\in\{1,\ldots,k\}$ we have a Rabin pair \vspace*{-0.3em}
$$\Big(F_i\times \B,I_i^{\ell_i}\times\{b\in \B\mid b(i)=\ell_i\}\Big)$$ \myspace\myspace

\begin{example}
In Example~\ref{ex:trans-fin} there is one pair and the degeneralization index is 2.
\end{example}

\begin{example}\label{ex:fair}
For a conjunction of three fairness constraints $\varphi=(\F\G a\vee\G\F b)\wedge (\F\G c \vee \G\F d)\wedge (\F\G e \vee \G\F f)$, the B\"uchi components $\I_i$'s of the equivalent GRP condition correspond to $\true,b,d,f,b\wedge d,b\wedge f,d\wedge f,b\wedge d\wedge f$. The degeneralization index is thus $|\B|=1\cdot1\cdot1\cdot1\cdot2\cdot2\cdot2\cdot3=24$. For four constraints, it is $1\cdot1^4\cdot2^6\cdot 3^4 \cdot 4=20736$. One can easily see the index grows doubly exponentially.
\end{example}

\myspace\myspace

\section{Probabilistic Model Checking}

\myspace

\input{krish_mdp}

\myspace\myspace

\subsection{Experimental results}
\label{sec:pMC:experiments}
In this section, we compare the performance of 

\myspaceb

\begin{description}
 \item[\L~]  the original \textsf{PRISM} with its implementation of \textsf{ltl2dstar} producing Rabin automata,
 \item[\D~]  \textsf{PRISM} with \textsf{Rabinizer}~\cite{atva} (our implementation of~\cite{cav}) producing DRW via \emph{optimized} degeneralization of DGRW, and
 \item[\GR] \textsf{PRISM} with \textsf{Rabinizer} producing DGRW and with the modified MEC checking step.
\end{description}

\myspaceb

We have performed a case study on the Pnueli-Zuck randomized mutual exclusion protocol~\cite{PZ86} implemented as a PRISM benchmark. We consider the protocol with 3, 4, and 5 participants. The sizes of the respective models are $s_3=2\,368$, $s_4=27\,600$, and $s_5=308\,800$ states. We have checked these models against several formulae illustrating the effect of the degeneralization index on the speed up of our method; see Table~\ref{table:mdp-experiments}.

In the first column, there are the formulae in the form of a \textsf{PRISM} query. We ask for a maximal/minimal value over all schedulers. Therefore, in the $P_{max}$ case, we create an automaton for the formula, whereas in the case of $P_{min}$ 
we create an automaton for its negation. The second column then states the number $i$ of participants, thus inducing the respective size $s_i$ of the model.

The next three columns depict the size of the product of the system and the automaton, for each of the \textbf{\L}, \textbf{\D}, \textbf{\GR} variants. The size is given as the ratio of the actual size and the respective $s_i$. The number then describes also the ``effective'' size of the automaton when taking the product. The next three columns display the total running times for model checking in each variant.

The last three columns illustrate the efficiency of our approach. The first column ${t_{\textbf{\D}}}/{t_{\textbf{\GR}}}$ states the time speed-up of the DGRW approach when compared to the corresponding degeneralization. The second column states the degeneralization index $|\B|$. The last column ${t_{\textbf{\L}}}/{t_{\textbf{\GR}}}$ then displays the overall speed-up of our approach to the original PRISM.

In the formulae, an atomic proposition $p_i=j$ denotes that the $i$th participant is in its state $j$. The processes start in state 0. In state 1 they want to enter the critical section. State 10 stands for being in the critical section. After leaving the critical section, the process re-enters state 0 again.

\newcommand{\equal}{\mathord{=}}
\newcommand{\nequal}{\mathord{\neq}}

\begin{table}[!th]
\vspace{-2em}
\caption{Experimental comparison of \textbf{\L}, \textbf{\D}, and \textbf{\GR} methods. All measurements performed on Intel i7 with 8 GB RAM. The sign ``$-$'' denotes either crash, out-of-memory, time-out after 30 minutes, or a ratio where one operand is $-$.}
\label{table:mdp-experiments}
\vspace*{-2em}
$$\begin{array}{|l|l|rrr|rrr|rc|r|}
\hline
\text{Formula}&\text{\#}&\frac{s_{\textbf{\L}}}{s_i}&\frac{s_{\textbf{\D}}}{s_i}&\frac{s_{\textbf{\GR}}}{s_i}& t_{\textbf{\L}}&t_{\textbf{\D}}&t_{\textbf{\GR}}&\frac{t_{\textbf{\D}}}{t_{\textbf{\GR}}}&|\B|&\frac{t_{\textbf{\L}}}{t_{\textbf{\GR}}}\\\hline
P_{max}=?[\G\F p_1\equal10   &3&4.1&2.6&1&1.2&0.4&0.2&2.2&3&6.8\\
\qquad\wedge\ \G\F p_2\equal10  &4&4.3&2.7&1&17.4&1.8&0.3&6.4&3&60.8\\
\qquad\wedge\ \G\F p_3\equal10]&5&4.4&2.7&1&257.5&15.2&0.6&26.7&3&447.9\\\hline
P_{max}=?[\G\F p_1\equal10\wedge\G\F p_2\equal10  &4&6&3.5&1&27.3&2.5&0.9&2.8&4&32.1\\
\qquad \wedge\ \G\F p_3\equal10 \wedge \G\F p_4\equal10]&5&6.2&3.6&1&408.5&17.8&0.9&20.4&4&471.2\\\hline
P_{min}=?[\G\F p_1\equal10\wedge\G\F p_2\equal10 &4&-&1&1&-&36.5&36.3&1&1&-\\
\qquad \wedge\ \G\F p_3\equal10 \wedge \G\F p_4\equal10]&5&-&1&1&-&610.6&607.2&1&1&-\\\hline
P_{max}=?[ (\G\F p_1 \equal 0 \vee \F\G p_2\nequal0)  &3&79.7&1.9&1&225.5&4.1&2.2&1.8&2&101.8\\
\qquad\wedge(\G\F p_2 \equal 0 \vee \F\G p_3\nequal0)]&4&-&1.9&1&-&61.7&29.2&2.1&2&-\\
&5&-&1.9&1&-&1007&479&2.1&2&-\\\hline
P_{max}=?[ (\G\F p_1 \equal 0 \vee \F\G p_1\nequal0)  &3&23.3&1.9&1&66.4&3.92&2.2&1.8&2&30.7\\
\qquad\wedge(\G\F p_2 \equal 0 \vee \F\G p_2\nequal0)]&4&23.3&1.9&1&551.5&61&28.2&2.2&2&19.6\\
&5&-&1.9&1&-&1002.7&463&2.2&2&-\\\hline
P_{max}=?[ (\G\F p_1 \equal 0 \vee \F\G p_2\nequal0)  &3&-&16.3&1&-&122.1&7.1&17.2&24&-\\
\qquad\wedge(\G\F p_2 \equal 0 \vee \F\G p_3\nequal0)    &4&-&-&1&-&-&75.6&-&24&-\\
\qquad\wedge(\G\F p_3 \equal 0 \vee \F\G p_1\nequal0)] &5&-&-&1&-&-&1219.5&-&24&-\\\hline
P_{max}=?[ (\G\F p_1 \equal 0 \vee \F\G p_1\nequal0)  &3&-&12&1&-&76.3&7.2&12&24&-\\
\qquad\wedge(\G\F p_2 \equal 0 \vee \F\G p_2\nequal0) &4&-&12.1&1&-&1335.6&78.9&19.6&24&-\\
\qquad\wedge(\G\F p_3 \equal 0 \vee \F\G p_3\nequal0)] &5&-&-&1&-&-&1267.6&-&24&-\\\hline
P_{min}=?[  (\G\F p_1\nequal 10\vee \G\F p_1\equal 0 \vee \F\G p_1\equal1) &3&2.1&1&1&1.2&0.9&0.8&1&1&1.5\\
\qquad\wedge \G\F p_1 \nequal0 \wedge\G\F p_1\equal1]&4&2.1&1&1&11.8&8.7&8.8&1&1&1.3\\
\qquad
&5&2.1&1&1&186.3&147.5&146.2&1&1&1.3\\\hline
P_{max}=?[(\G p_1\nequal10 \vee \G p_2\nequal10 \vee \G p_3\nequal10)&3&-&32&5.9&-&405&80.1&5.1&8&-\\
~~\wedge (\F\G p_1\nequal 1 \vee\G\F p_2=1\vee\G\F p_3=1)&4&-&-&6.4&-&-&703.5&-&8&-\\
~~\wedge (\F\G p_2\nequal 1 \vee\G\F p_1=1\vee\G\F p_3=1)&5&-&-&-&-&-&-&-&8&-\\\hline
P_{min}=?[(\F \G p_1\nequal0 \vee \F \G p_2\nequal0 \vee \G \F p_3\equal0)&3&55.9&4.7&1&289.7&12.6&3.4&3.7&12&84.3\\
~~\vee (\F\G p_1\nequal10 \wedge \G\F p_2=10 \wedge \G\F p_3=10)&4&-&4.6&1&-&194.5&33.2&5.9&12&-\\
&5&-&-&1&-&-&543&-&12&-\\
\hline
\end{array}$$
\vspace*{-3em}
\end{table}

Formulae 1 to 3 illustrate the effect of $|\B|$ on the ratio of sizes of the product in the \textbf{\D} and \textbf{\GR} cases, see $\frac{s_{\textbf{\D}}}{s_i}$, and ratio of the required times. 
The theoretical prediction is that ${s_{\textbf{\D}}}/{s_{\textbf{\GR}}}=|\B|$. Nevertheless, due to optimizations done in the degeneralization process, the first is often slightly smaller than the second one, see columns $\frac{s_{\textbf{\D}}}{s_i}$ and $\B$. (Note that $s_{\textbf{\GR}}/s_i$ is 1 for ``infinitary'' formulae.) 
For the same reason, $\frac{t_{\textbf{\D}}}{t_{\textbf{\GR}}}$ is often smaller than $|\B|$. However, with the growing size of the systems it gets bigger hence the saving factor is larger for larger systems, as discussed in the previous section.

Formulae 4 to 7 illustrate the doubly exponential growth of $|\B|$ and its impact on systems of different sizes. The DGRW approach (\textbf{\GR} method) is often the only way to create the product at all.

Formula 8 is a Streett condition showing the approach still performs competitively. Formulae 9 and 10 combine Rabin and Streett condition requiring both 
big Rabin automata and big Streett automata. Even in this case, the method scales well.
Further, Formula 9 contains non-infinitary behaviour, e.g.~$\G p_1\mathord{\neq}10$. Therefore, the DGRW is of size greater than 1, and thus also the product is bigger as can be seen in the $s_{\textbf{\GR}}/s_i$ column.

\myspace

\section{Synthesis}

\myspace

In this section, we show how generalized Rabin pairs can be used to speed up the computation of a winning strategy in an LTL(\F,\G) game and thus to speed up LTL(\F,\G) synthesis.
A game is defined like an MDP, but with the stochastic vertices replaced by vertices of an adversarial player.

\begin{definition}
A \emph{game} $\mathcal M=(V, E,(V_0,V_1))$ 
consists of a finite directed {\em game graph} $(V,E)$ and a partition $(V_0,V_1)$ 
of the \emph{finite} set $V$ of vertices into player-0 vertices ($V_0$) and
player-1 vertices ($V_1$). 
\end{definition}

An \emph{LTL game} is a game together with an LTL formula with vertices as atomic propositions. Similarly, a \emph{Rabin game} and a \emph{game with GRP condition (GRP game)} is a game with a set of Rabin pairs, or a set of generalized Rabin pairs, respectively. 

A \emph{strategy} is a function $V^*\to E$ assigning to each \emph{history} an outgoing edge of its last vertex. A play conforming to the strategy $f$ of Player 0 is any infinite sequence $v_0v_1\cdots$ satisfying $v_{i+1}=f(v_0\cdots v_i)$ whenever $v_i\in V_0$, and just $(v_i,v_{i+1})\in E$ otherwise. Player 0 has a \emph{winning strategy}, if there is a strategy $f$ such that all plays conforming to $f$ of Player 0 satisfy the LTL formula, Rabin condition or GRP condition, depending on the type of the game. For further details, we refer to e.g.~\cite{PP}.

One way to solve an LTL game is to make a product of the game arena with the DRW corresponding to the LTL formula, yielding a Rabin game. The current fastest solution of Rabin games works in time $\mathcal O(mn^{k+1}kk!)$~\cite{PP}, where $n=|V|,m=|E|$ and $k$ is the number of pairs. Since $n$ is doubly exponential and $k$ singly exponential in the size of the formula, this leads to a doubly exponential algorithm. And indeed, the problem of LTL synthesis is 2-EXPTIME-complete~\cite{DBLP:conf/icalp/PnueliR89}. 

Similarly as for model checking of probabilistic systems, we investigate what happens (1) if we replace the translation to Rabin automata by our new translation and (2) if we employ DGRW instead. The latter leads to the problem of GRP games. In order to solve them, we extend the methods to solve Rabin and Streett games of~\cite{PP}. 

We show that solving a GRP game is faster than first degeneralizing them and then solving the resulting Rabin game. The induced speed-up factor is $|\B|^k$. 
In the following two subsections we show how to solve GRP games and analyze the complexity. The subsequent section reports on experimental results. 
\myspace

\subsection{Generalized Rabin ranking}

We shall compute a ranking of each vertex, which intuitively states how far from winning we are. The existence of winning strategy is then equivalent to
the existence of a ranking where Player 0 can always choose a successor of the current vertex with smaller ranking, i.e.~closer to fulfilling the goal.

Let $(V,E,(V_0,V_1))$ be a game, $\{(F_1,\I_1),\ldots,(F_k,\I_k)\}$ a GRP condition with the corresponding degeneralization domain $\B$. Further, let $n:=|V|$ and denote the set of permutations over a set $S$ by $S!$.

\begin{definition}
A \emph{ranking} is a function $r:V\times B\to R$ where $R$ is the \emph{ranking domain}
$\{1,\ldots,k\}!\times\{0,\ldots,n\}^{k+1}\cup\{\infty\}$.
\end{definition}
The ranking $r(v,\elemB)$ gives information important in the situation when we are in vertex $v$ and are waiting for a visit of $\I_i^{\elemB(i)}$ for each $i$ given by $\elemB\in\B$. As time passes the ranking should decrease. To capture this, we define the following functions.
\begin{definition}\label{def:nxt}
For a ranking $r$ and given $v\in V$ and $\elemB\in \B$, we define $\nxt_v:\B\to\B$
$$\nxt_v(\elemB)(i)=
\begin{cases}
\elemB(i) &\text{if }v\notin\I_i^{\elemB(i)}\\
\elemB(i)\oplus1 &\text{if }v\in\I_i^{\elemB(i)}
\end{cases}
$$
and $\nxt:V\times\B\to R$ \myspace
$$\nxt(v,\elemB)=
\begin{cases}
\min_{(v,w)\in E} r(w,\nxt_v(\elemB)) &\text{if }v\in V_0\\
\max_{(v,w)\in E} r(w,\nxt_v(\elemB)) &\text{if }v\in V_1\\
\end{cases}
$$
where the order on $(\pi_1\cdots\pi_k,w_0w_1\cdots w_k)\in R$ is given by the lexicographic order $>_{lg}$ on $w_0\pi_1w_1\pi_2w_2\cdots\pi_kw_k$ and $\infty$ being the greatest element.
\end{definition}
Intuitively, the ranking $r(v,\elemB)=(\pi_1\cdots\pi_k,w_0w_1\cdots w_k)$ is intended to bear the following information. The permutation $\pi$ states the importance of the pairs. The pair $(F_{\pi_1},\I_{\pi_1})$ is the most important, hence we are not allowed to visit $F_{\pi_1}$ and we desire to either visit $\I_{\pi_1}$, or not visit $F_{\pi_2}$ and visit $\I_{\pi_2}$ and so on. If some important $F_i$ is visited it becomes less important. The importance can be freely changed only finitely many ($i_0$) times. Otherwise, only less important pairs can be permuted if a more important pair makes good progress. Further, $w_i$ measures the worst possible number of steps until visiting $\I_{\pi_i}$. This intended meaning is formalized in the following notion of good rankings.
\begin{definition}\label{def:good}
A ranking $r$ is \emph{good} if for every $v\in V,\elemB\in\B$ with $r(v,\elemB)\neq\infty$ we have $r(v,\elemB)>_{v,\elemB} \nxt(v,\elemB)$.

We define $(\pi_1\cdots\pi_k,w_0w_1\cdots w_k) >_{v,\elemB} (\pi_1'\cdots\pi_k',w_0'w_1'\cdots w_k')$ if either $w_0>w_0'$, or $w_0=w_0'$ with $>_{v,\elemB}^1$ hold. Recursively, $>_{v,\elemB}^\ell$ holds if one of the following holds:
\begin{itemize}
 \item $\pi_\ell>\pi_\ell'$
 \item $\pi_\ell=\pi_\ell',\;v\not\models F_{\pi_\ell}$ and $w_\ell>w_\ell'$
 \item $\pi_\ell=\pi_\ell',\;v\not\models F_{\pi_\ell}$ and $v\models \I_{\pi_\ell}^{\elemB(\pi_\ell)}$
 \item $\pi_\ell=\pi_\ell',\;v\not\models F_{\pi_\ell}$ and $w_\ell=w_\ell'$ and $>_{v,\elemB}^{\ell+1}$ holds (where $>_{v,\elemB}^{k+1}$ never holds)
\end{itemize}
Moreover, if one of the first three cases holds, we say that $\succ_{v,\elemB}^\ell$ holds.
\end{definition}
Intuitively, $>$ means the second element is closer to the next milestone and $\succ^\ell$, moreover, that it is so because of the first $\ell$ pairs in the permutation.

Similarly to~\cite{PP}, we obtain the following correctness of the construction.
Note that for $|\B|=1$, the definitions of the ranking here and the \emph{Rabin ranking} of~\cite{PP} coincide. Further, the extension with $|\B|>1$ bears some similarities with the Streett ranking of~\cite{PP}.

\begin{theorem}
For every vertex $v$, Player 0 has a winning strategy from $v$ if and only if there is a good ranking $r$ and $\elemB\in\B$ with $r(v,\elemB)\neq\infty$. 
\end{theorem}

\myspace

\subsection{A fixpoint algorithm}


In this section, we show how to compute the smallest good ranking and thus solve the GRP game.
Consider a lattice of rankings ordered component-wise, i.e. $r_1>_c r_2$ if for every $v\in V$ and $\elemB\in B$, we have $r_1(v,\elemB)>_{lg} r_2(v,\elemB)$. This induces a complete lattice. The minimal good ranking is then a least fixpoint of the operator $\lift$ on rankings given by:
$$\lift(r)(v,\elemB)=\max\big\{r(v,\elemB),\min\{x\mid x>_{v,\elemB}\nxt(v,\elemB)\}\big\}$$
where the optima are considered w.r.t. ${>_{lg}}$. Intuitively, if Player 0 cannot choose a successor smaller than the current vertex (or all successors of a Player 1 vertex are greater), the ranking of the current vertex must rise so that it is greater.

\begin{theorem}
The smallest good ranking can be computed in time $\mathcal O(mn^{k+1}kk!\cdot |\B|)$ and space $(nk\cdot|\B|)$.
\end{theorem}
\begin{proof}
The lifting operator can be implemented similarly as in~\cite{PP}. With every change, the affected predecessors to be updated are put in a worklist, thus working in time $\mathcal O(k\cdot\text{out-deg}(v))$. Since every element can be lifted at most $|R|$-times, the total time is $ \mathcal O(\sum_{v\in V}\sum_{\elemB\in\B}k\cdot\text{out-deg}(v)\cdot|R|)=|\B|km\cdot n^{k+1}k!$. The space required to store the current ranking is $ \mathcal O(\sum_{v\in V}\sum_{\elemB\in\B}k)=n\cdot|B|\cdot k$. \qed
\end{proof}

We now compare our solution to the one that would solve the degeneralized Rabin game. The number of vertices of the degeneralized Rabin game is $|\B|$ times greater. Hence the time needed is multiplied by a factor $|\B|^{k+2}$, 
instead of $|\B|$ in the case of a GRP game. Therefore, our approach speeds up by a factor of $|\B|^{k+1}$, while the space requirements are the same in both cases, namely $\mathcal O(nk\cdot|\B|)$.

\begin{example}
A conjunction of two fairness constraints of example~\ref{ex:fair} corresponds to $|\B|=2$ and $k=4$, hence we save by a factor of $2^4=16$. A conjunction of three fairness constraints corresponds to $|\B|=24$ and $k=8$, hence we accelerate $24^8\approx10^{11}$ times.
\end{example}

Further, let us note that the computation can be implemented recursively as in~\cite{PP}. The winning set is $\mu Z.\;\mathfrak{GR}(\mathcal{GR},\true,\pred Z)$ where $\mathfrak{GR}(\emptyset,\varphi,W)= W$,
\begin{align*}
 \mathfrak{GR}(\mathcal{GR},\varphi,W)&= \bigvee_{i\in[1..k]}\nu Y.\bigwedge_{j\in[1..|\I_i|]} \mu X.\; \mathfrak{GR}\Big(\mathcal{GR}\setminus\{(F_i,\I_i)\},\varphi\wedge\neg F_i, \\
& ~~~~~~~~~~~~~~~~~~~~~~~~~~~~
 W\vee(\varphi\wedge\neg F_i\wedge I_i^j\wedge \pred Y) \vee (\varphi\wedge\neg F\wedge \pred X) \Big)
\end{align*}
$\pred\varphi=\{u\in V_0\mid\exists (u,v)\in E: v\models\varphi\}\cup \{u\in V_1\mid\forall (u,v)\in E: v\models\varphi\}$ and $\mu$ and $\nu$ denote the least and greatest fixpoints, respectively. The formula then provides a succinct description of a symbolic algorithm.

\myspace

\subsection{Experimental Evaluation}
\label{sec:synthesis:experiments}

\myspaceb

Reusing the notation of Section~\ref{sec:pMC:experiments}, we compare the performance of the methods for solving LTL games. We build and solve a Rabin game using
\begin{description}
 \item[\L~] \hspace*{4.3mm}\textsf{ltl2dstar} producing DRW (from LTL formulae),
 \item[\D~] \quad\textsf{Rabinizer} producing DRW, and
 \item[\GR] \hspace*{1.7mm}\textsf{Rabinizer} producing DGRW.
\end{description}

We illustrate the methods on three different games and three LTL formulae; see
Table~\ref{table:game-experiments}. The games contain 3 resp. 6 resp. 9 vertices. 
Similarly to Section~\ref{sec:pMC:experiments}, $s_i$ denotes the number of vertices in the $i$th arena,
$s_{\textbf{\L}}, s_{\textbf{\D}}, s_{\textbf{\GR}}$ the number of vertices in the resulting games for the three
methods, and $t_{\textbf{\L}}, t_{\textbf{\D}}, t_{\textbf{\GR}}$ the respective running times.

Formula 1 allows for a winning strategy and the smallest ranking is relatively small, hence computed quite fast. 
Formula 2, on the other hand, only allows for larger rankings.
Hence the computation takes longer, but also because in \textbf{\L} and \textbf{\D} cases the automata are larger than for formula 1. While for \textbf{\L} and \textbf{\D}, the product is usually too big, there is a chance to find small rankings in \textbf{\GR} fast.
While for e.g.~$\F\G (a \vee \neg b \vee c)$, the automata and games would be the same for all three methods and the solution would only take less than a second, the more complex formulae 1 and 2 show clearly the speed up. 

\begin{table}[!th]
\caption{Experimental comparison of \textbf{\L}, \textbf{\D}, and \textbf{\GR} methods for solving LTL games. 
Again the sign ``$-$'' denotes either crash, out-of-memory, time-out after 30 minutes, or a ratio where one operand is $-$.}
\label{table:game-experiments}
\vspace{-2em}
$$\begin{array}{|l|l|rrr|rrr|rc|r|}
\hline
\text{Formula}                                                 &s_i&\frac{s_{\textbf{\L}}}{s_i}&\frac{s_{\textbf{\D}}}{s_i}&\frac{s_{\textbf{\GR}}}{s_i}& t_{\textbf{\L}}&t_{\textbf{\D}}&t_{\textbf{\GR}}   &\frac{t_{\textbf{\D}}}{t_{\textbf{\GR}}}&|\B|&\frac{t_{\textbf{\L}}}{t_{\textbf{\GR}}}\\\hline

(\G\F a \wedge \G\F b \wedge \G\F c)                           &3  & 22                       & 7.3                      &                     4   &          63.2 &           1.6&    1.1       &1.4 &9   &                          48.2 \\
\qquad\vee (\G\F \neg a \wedge  \G\F \neg b \wedge  \G\F \neg c)   &6  & 21.3                     & 7.3                      &                     3.7 &         878.6 &          14.1&    7.3      &2 &9  &                          130.3  \\
\qquad                                                         &9  & 20.6                     & 7                        &                     3.6 &             - &           54.8&    31.3    &1.8 &9       &                       -\\\hline

(\G\F a \vee \F\G  b)  \wedge (\G\F c \vee \G\F \neg a)                                       &3  & 21                       & 10                       &                     4   &             - &                 117.5&    12     &9.8 &6                              &-\\
\qquad \wedge (\G\F c \vee \G\F \neg b)                             &6  & 16.2                     & 9.2                      &                     3.7 &             - &                -    &    196.7   &-   &6                              &-\\
                         &9  & 17.6                     & 9.2                     &                     3.6 &             - &               -    &    1017.8   &-   &6                              &-\\\hline
\end{array}$$
\vspace{-2em}
\end{table}





\section{Conclusions}
\input{conclusion}


\bibliographystyle{alpha} 
\bibliography{refs-short}

\newpage
\appendix
\input{appendix}

\newpage

\input{app-fig}

\end{document}

%% file: intro.tex
\myspace\myspace

\begin{abstract}
The model-checking problem for probabilistic systems crucially relies on 
the translation of LTL to deterministic Rabin automata (DRW). 
Our recent Safraless translation~\cite{cav,atva} for the LTL(\F,\G) fragment 
produces smaller automata as compared to the traditional approach.
In this work, instead of DRW we consider deterministic automata with 
acceptance condition given as disjunction of generalized Rabin pairs (DGRW).
The Safraless translation of LTL(\F,\G) formulas to DGRW results in smaller 
automata as compared to DRW.
We present algorithms for probabilistic model-checking as well as game solving
for DGRW conditions.
Our new algorithms lead to improvement both in terms of theoretical bounds 
as well as practical evaluation.
We compare PRISM with and without our new translation, and show that the new
translation leads to significant improvements. 
\end{abstract}

\myspace\myspace\myspace

\section{Introduction}

\myspace

\noindent{\em Logic for $\omega$-regular properties.}
The class of $\omega$-regular languages generalizes regular languages
to infinite strings and provides a robust specification language to 
express all properties used in verification and synthesis.
The most convenient way to describe specifications is through logic,
as logics provide a concise and intuitive formalism to express
properties with very precise semantics. 
The linear-time temporal logic (LTL)~\cite{DBLP:conf/focs/Pnueli77} 
is the de-facto logic to express linear time $\omega$-regular properties 
in verification and synthesis.

\smallskip\noindent{\em Deterministic $\omega$-automata.}
For model-checking purposes, LTL formulas can be converted to 
nondeterministic B\"uchi automata (NBW)~\cite{DBLP:journals/jcss/VW86}, and then the problem 
reduces to checking emptiness of the intersection of two NBWs 
(representing the system and the negation of the specification, respectively).
However, for two very important problems deterministic automata 
are used, namely, (1)~the synthesis problem~\cite{Church62,DBLP:conf/icalp/PnueliR89}; and 
(2)~the model-checking problem for probabilistic systems or 
Markov decision processes (MDPs)~\cite{DBLP:books/daglib/0020348} which has
a wide range of applications from randomized communication, to security 
protocols, to biological systems.
The standard approach is to translate LTL to NBW~\cite{DBLP:journals/jcss/VW86}, and then 
convert the NBW to a deterministic automata with Rabin acceptance 
condition (DRW) using Safra's determinization procedure~\cite{DBLP:conf/focs/Safra88} (or using 
a recent improvement of Piterman~\cite{DBLP:conf/lics/Piterman06}).

\smallskip\noindent{\em Avoiding Safra's construction.} 
The key bottleneck of the standard approach in practice is Safra's 
determinization procedure which is difficult to implement due to 
the complicated state space and data structures associated with the
construction~\cite{DBLP:conf/sofsem/Kupferman12}.
As a consequence several alternative approaches have been proposed, and
the most prominent ones are as follows.
The first approach is the Safraless approach. One can reduce the synthesis 
problem to emptiness of nondeterministic B\"uchi tree automata~\cite{DBLP:conf/focs/KupfermanV05};
it has been implemented with considerable success in~\cite{DBLP:conf/fmcad/JobstmannB06}. For probabilistic model checking other constructions can be also used, however, all of them are exponential~\cite{DBLP:conf/focs/Vardi85,DBLP:journals/jacm/CourcoubetisY95}.
The second approach is to use heuristic to improve Safra's determinization
procedure~\cite{DBLP:journals/tcs/KleinB06,DBLP:conf/wia/KleinB07} which has led to the tool ltl2dstar~\cite{ltl2dstar}.
The third approach is to consider fragments of LTL. 
In~\cite{DBLP:journals/tocl/AlurT04} several simple fragments of LTL were proposed that allow
much simpler (single exponential as compared to the general double exponential)
translations to deterministic automata.
The generalized reactivity(1) fragment of LTL (called GR(1))  was introduced 
in~\cite{DBLP:conf/vmcai/PitermanPS06} and a cubic time symbolic representation of an equivalent automaton was 
presented. The approach has been implemented in the ANZU tool~\cite{DBLP:conf/cav/JobstmannGWB07}.
Recently, the $(\F,\G)$-fragment of LTL, that uses boolean operations and 
only $\F$ (eventually or in future) and $\G$ (always or globally) as temporal operators, was considered
and a simple and direct translation to deterministic Rabin automata (DRW) was 
presented~\cite{cav}. Not only it covers all fragments of~\cite{DBLP:journals/tocl/AlurT04}, but it can also express all complex fairness constraints, which are widely used in verification.

\smallskip\noindent{\em Probabilistic model-checking.} 
Despite several approaches to avoid Safra's determinization, for probabilistic 
model-checking the deterministic automata are still necessary.
Since probabilistic model-checkers handle linear arithmetic, they do not benefit 
from the symbolic methods of~\cite{DBLP:conf/vmcai/PitermanPS06,DBLP:conf/vmcai/MorgensternS08} or from the tree automata approach.
The approach for probabilistic model-checking has been to explicitly 
construct a DRW from the LTL formula. 
The most prominent probabilistic model-checker PRISM~\cite{prism} implements
the ltl2dstar approach.

\smallskip\noindent{\em Our results.} In this work, we focus on the $(\F,\G)$-fragment
of LTL.
Instead of the traditional approach of translation to DRW we propose a translation 
to deterministic automata with \emph{generalized Rabin pairs}.
We present probabilistic model-checking as well as symbolic game solving algorithms
for the new class of conditions which lead to both theoretical as well as 
significant practical improvements.
The details of our contributions are as follows.

\myspacec

\begin{enumerate}

\item A Rabin pair consists of the conjunction of a B\"uchi (always eventually) 
and a coB\"uchi (eventually always) condition, and a Rabin condition is a 
disjunction of Rabin pairs. 
A generalized Rabin pair is the conjunction of conjunctions of B\"uchi 
conditions and conjunctions of coB\"uchi conditions.
However, as conjunctions of coB\"uchi conditions is again a coB\"uchi condition,
a generalized Rabin pair is the conjunction of a coB\"uchi condition and 
conjunction of B\"uchi conditions.\footnote{Note that our condition (disjunction of generalized Rabin pairs) is very different from both generalized Rabin conditions (conjunction of Rabin conditions) and 
the generalized Rabin(1) condition of~\cite{DBLP:conf/nfm/Ehlers11}, which considers a set of assumptions and guarantees where each assumption and guarantee consists of \emph{one} Rabin pair.
Syntactically, disjunction of generalized Rabin pairs condition is 
$\bigvee_{i}( \F\G a_i \wedge \bigwedge_j \G\F b_{ij})$,
whereas generalized Rabin condition is $\bigwedge_j (\bigvee_i (\F\G a_{ij} \wedge \G\F b_{ij}))$,
and generalized Rabin(1) condition is $(\bigwedge_i (\F\G a_i \wedge \G\F b_i) \Rightarrow \bigwedge_j (\F\G a_j \wedge \G\F b_j))$.}
We consider deterministic automata where the acceptance condition is a disjunction of 
generalized Rabin pairs (and call them DGRW). 
The $(\F,\G)$-fragment of LTL admits a direct and algorithmically simple translation to 
DGRW~\cite{cav} and we consider DGRW for probabilistic model-checking and synthesis.
The direct translation of LTL(\F,\G) could be done to a compact deterministic 
automaton with a Muller condition, however, the explicit representation of 
the Muller condition is typically huge and not algorithmically efficient, 
and thus reduction to deterministic Rabin automata was performed (with a 
blow-up) since Rabin conditions admit efficient algorithmic analysis. 
We show that DGRW allow both for a very compact translation of the $(\F,\G)$-fragment 
of LTL as well as efficient algorithmic analysis.
The direct translation of LTL(\F,\G) to DGRW has the same number of states as 
for a general Muller condition.
For many formulae expressing e.g.~fairness-like conditions the translation to DGRW is 
significantly more compact than the previous ltl2dstar approach. 
For example, for a conjunction of three strong fairness constraints,
ltl2dstar produces a DRW with more than a million states, translation to DRW via DGRW requires 469 states, 
and the corresponding DGRW has only 64 states.

\item One approach for probabilistic model-checking and synthesis for DGRW 
would be to first convert them to DRW, and then use the standard algorithms. 
Instead we present direct algorithms for DGRW that avoids the translation to 
DRW both for probabilistic model-checking and game solving.
The direct algorithms lead to both theoretical and practical improvements.
For example, consider the disjunctions of $k$ generalized Rabin pairs 
such that in each pair there is a conjunction of a coB\"uchi condition and 
conjunctions of $j$ B\"uchi conditions.
Our direct algorithms for probabilistic model-checking as well as game solving 
is more efficient by a multiplicative factor of $j^k$ and $j^{k^2+k}$ as compared to the 
approach of translation to DRW for probabilistic model checking and game solving,
respectively.
Moreover, we also present symbolic algorithms for game solving for DGRW 
conditions. 

\item We have implemented our approach for probabilistic model checking in PRISM,
and the experimental results show that as compared to the existing implementation
of PRISM with ltl2dstar our approach results in improvement of order of magnitude.
Moreover, the results for games confirm that the speed up is even greater than for probabilistic model checking.

\end{enumerate}

%% file: krish_mdp.tex
\newcommand{\trans}{\delta}
\newcommand{\MEC}{\mathsf{MEC}}
\newcommand{\LP}{\mathsf{LP}}

In this section, we show how automata with generalized Rabin pairs can 
significantly speed up model checking of Markov decision processes 
(i.e., probabilistic model checking).
For example, for the fairness constraints of the type mentioned in 
Example~\ref{ex:fair} the speed-up is by a factor that is doubly exponential. 
Although there are specialized algorithms for checking properties under 
strong fairness constraints (implemented in PRISM), our approach is general 
and speeds up for a wide class of constraints. 
The combinations (conjunctions, disjunctions) of properties not expressible 
by small Rabin automata (and/or Streett automata) are infeasible for the 
traditional approach, while we show that automata with generalized Rabin pairs 
often allow for efficient model checking.
First, we present the theoretical model-checking algorithm for the new type 
of automata and the theoretical bounds for savings. 
Second, we illustrate the effectiveness of the approach experimentally.

\myspace

\subsection{Model checking using generalized Rabin pairs}

We start with the definitions of Markov decision processes (MDPs), and present the 
model-checking algorithms. For a finite set $V$, let $\mathrm{Distr}(V)$ denote the set of probability distributions on $V$.

\begin{definition}[MDP and MEC]
A \emph{Markov decision process (MDP)} $\mathcal M=(V, E,(V_0,V_P),\trans)$ 
consists of a finite directed {\em MDP graph} $(V,E)$, a partition $(V_0,V_P)$ 
of the \emph{finite} set $V$ of vertices into player-0 vertices ($V_0$) and
probabilistic vertices ($V_P$), and a probabilistic transition function 
$\trans$: $V_P \rightarrow  \mathrm{Distr}(V)$
such that for all vertices $u \in V_P$ and $v \in V$  we have $(u,v) \in E$ iff $\trans(u)(v)>0$.

An \emph{end-component} $U$ of an MDP is a set of its vertices such that (i)~the subgraph induced by $U$ is 
strongly connected and (ii)~for each edge $(u,v) \in E$, if $u \in U \cap V_P$, then $v \in U$
(i.e., no probabilistic edge leaves $U$).

A \emph{maximal end-component (MEC)} is an end-component that is maximal w.r.t.~to
the inclusion ordering.
\end{definition}

If $U_1$ and $U_2$ are two end-components and $U_1 \cap U_2 \neq \emptyset$, 
then $U_1 \cup U_2$ is also an end-component. Therefore, every MDP induces a unique set of its MECs, called \emph{MEC decomposition}.

For precise definition of semantics of MDPs we refer to~\cite{Puterman:book}.
Note that MDPs are also defined in an equivalent way in literature with a 
set of actions such that every vertex and choice of action determines the
probability distribution over the successor states; the choice of 
actions corresponds to the choice of edges at player-0 vertices of our 
definition.

The standard model-checking algorithm for MDPs proceeds in several steps. Given an MDP $\mathcal M$ 
and an LTL formula $\varphi$
\myspaceb
\begin{enumerate}
 \item compute a deterministic automaton $\mathcal A$ recognizing the language of $\varphi$,
 \item compute the product $\product=\mathcal M\times\mathcal A$,
 \item solve the product MDP $\product$.
\end{enumerate}
\myspaceb
The algorithm is generic for all types of deterministic $\omega$-automata $\mathcal A$. The leading probabilistic model checker \textsf{PRISM}~\cite{prism} re-implements \textsf{ltl2dstar}~\cite{ltl2dstar} that transforms $\varphi$ into a deterministic \emph{Rabin} automaton. This approach employs Safra's determinization and thus despite many optimization often results in an unnecessarily big automaton. 

There are two ways to fight the problem. Firstly, one can strive for smaller Rabin automata. Secondly, one can employ other types of $\omega$-automata. As to the former, we have plugged our implementation \textsf{Rabinizer}~\cite{atva} of the approach~\cite{cav} into \textsf{PRISM}, which already results in considerable improvement. For the latter, Example~\ref{ex:trans-infin} shows that Muller automata can be smaller than Rabin automata. However, explicit representation of Muller acceptance conditions is typically huge. Hence the third step to solve the product MDP would be too expensive. Therefore, we propose to use automata with generalized Rabin pairs.

On the one hand, DGRW often have small state space after translation. Actually, it is the same as the state space of the intermediate Muller automaton of~\cite{cav}. Compared to the corresponding naively degeneralized DRW it is $|\B|$ times smaller (one can still perform some optimizations in the degeneralization process, see the experimental results). 

On the other hand, as we show below the acceptance condition is still algorithmically efficient to handle. 
We now present the steps to solve the product MDP for a GRP acceptance condition, i.e.~a disjunction of
generalized Rabin pairs.
Consider an MDP with $k$ generalized Rabin pairs  $(F_i,\{I_i^1,\ldots,I_i^{\ell_i}\})$, for $i=1,2,\ldots,k$.
The steps of the computation are as follows:
\begin{enumerate}
\item For $i=1,2,\ldots, k$;
\begin{enumerate}
\item Remove the set of states $F_i$ from the MDP.
\item Compute the MEC decomposition.
\item If a MEC $C$ has a non-empty 
intersection with each $I_i^j$, for $j=1,2,\ldots,\ell_i$,
then include $C$ as a winning MEC.
\item let $W_i$ be the union of winning MECs (for the $i$th pair).
\end{enumerate}
\item Let $W$ be the union of $W_i$, i.e.~$W=\bigcup_{i=1}^k W_i$.
\item The solution (or optimal value of the product MDP) is the 
maximal probability to reach the set $W$.
\end{enumerate}
Given an MDP with $n$ vertices and $m$ edges, let $\MEC(n,m)$ denote the
complexity of computing the MEC decomposition; and
$\LP(n,m)$ denotes the complexity to solve linear-programming solution
with $m$ constraints over $n$ variables.

\begin{theorem}
Given an MDP with $n$ vertices and $m$ edges with $k$ generalized Rabin pairs 
$(F_i,\{I_i^1,\ldots,I_i^{\ell_i}\})$, for $i=1,2,\ldots,k$,
the solution can be achieved in time $\mathcal O(k \cdot \MEC(n,m) + n \cdot \sum_{i=1}^k \ell_i) + 
\mathcal O(\LP(n,m))$.
\end{theorem}

\begin{remark}
The best known complexity to solve MDPs with Rabin conditions 
of $k$ pairs require time $\mathcal O(k \cdot \MEC(n,m)) + \mathcal O(\LP(n,m))$ time~\cite{luca-thesis}.
Thus degeneralization of generalized Rabin pairs to Rabin conditions 
and solving MDPs would require time $\mathcal O(k \cdot \MEC(|\B|\cdot n,|\B| \cdot m)) 
+ \mathcal O(\LP(|\B|\cdot n,|\B|\cdot m))$ time.
The current best known algorithms for maximal end-component decomposition 
require at least $\mathcal O(m \cdot n^{2/3})$ time~\cite{CH11}, and the simplest algorithms
that are typically implemented require $\mathcal O(n \cdot m)$ time.
Thus our approach is more efficient at least by a factor of $B^{5/3}$ (given
the current best known algorithms), and
even if both maximal end-component decomposition and linear-programming
can be solved in linear time, our approach leads to a speed-up by a 
factor of $|\B|$, i.e.~exponential in $\mathcal O(k)$ the number of non-trivially generalized Rabin pairs.
In general if $\beta\geq 1$ is the sum of the exponents required to
solve the MEC decomposition (resp.\ linear-programming),
then our approach is better by a factor of $|\B|^\beta$. 
\end{remark}


\begin{example}
A Rabin automaton for $n$ constraints of Example~\ref{ex:fair} is of doubly exponential size, which is also the factor by which the product and thus the running time grows. However, as the formula is ``infinitary'' (see end of Section~\ref{ssec:dgrw}), the state space of the generalized automaton is $2^{Ap}$ and the product is of the very same size as the original system since the automaton only monitors the current labelling of the state.
\end{example}


%% file: conclusion.tex
In this work we considered the translation of the LTL(\F,\G) fragment to 
deterministic $\omega$-automata that is necessary for probabilistic model
checking as well as synthesis.
The direct translation to deterministic Muller automata gives a compact 
automata but the explicit representation of the Muller condition is huge
and not algorithmically amenable.
In contrast to the traditional approach of translation to deterministic Rabin 
automata that admits efficient algorithms but incurs a blow-up in translation, 
we consider deterministic automata with generalized Rabin pairs (DGRW).
The translation to DGRW produces the same compact automata as for Muller
conditions.
We presented efficient algorithms for probabilistic model checking and game 
solving with DGRW conditions which shows that the blow-up of translation 
to Rabin automata is unnecessary.
Our results establish that DGRW conditions provide the convenient formalism
that allows both for compact automata as well as efficient algorithms.
We have implemented our approach in PRISM, and experimental results show 
a huge improvement over the existing methods.
Two interesting directions of future works are (1)~extend our approach to 
LTL with the \U (until) and the \X (next) operators; and 
(2)~consider symbolic computation and Long's acceleration of fixpoint 
computation (on the recursive algorithm), instead of the ranking function based
algorithm for games, and compare the efficiency of both the approaches.


%% file: appendix.tex
\section{Proof of Theorem 17 (correctness of the ranking)}

\subsection{Soundness}\label{ssec:sound}

\begin{lemma}
For every good ranking $r$ and vertex $v$ with $r(v,\elemB)\neq\infty$ for some $\elemB\in\B$, Player 0 has a winning strategy from $v$.
\end{lemma}
\begin{proof}
We construct a strategy with memory $\B$ and memory update $(v,\elemB)\mapsto\nxt_v(\elemB)$. When in vertex $v$ with memory $\elemB$, the strategy chooses a successor $v'$ for which $r(v',\nxt_v(\elemB))=\nxt(v,\elemB)$, i.e.~with the lowest admissible ranking. We prove it is winning from $v$.

Consider an infinite play $v^0v^1\cdots$ conforming to the strategy and $\elemB^0\elemB^1\cdots$ the corresponding memories and $r^0r^1\cdots$ the corresponding ranks $r^i=r(v^i,\elemB^i)=(\pi^i_1\cdots\pi^i_k,w^i_0\cdots w^i_k)$. By definitions~\ref{def:good} and~\ref{def:nxt}, $r^i>_{v^i,\elemB^i}r^{i+1}$ for all $i$.

Let $\ell$ be the smallest number for which $r^m\succ_{v^m,\elemB^m}^\ell r^{m+1}$ for infinitely many $m$. Then for almost all $i$
\begin{itemize}
 \item $r^i$ are the same on the first $\ell$ elements of both components, i.e. on each of $\pi^i_1$ to $\pi^i_\ell$ and $w^i_0$ to $w^i_{\ell-1}$, we denote the repetitive $\pi^i_\ell$ by $\mathit{win}$,
 \item thus also $v^i\not\models F_{\mathit{win}}$, and
 \item $\elemB^i$ are the same on the first $\ell-1$ elements,
\end{itemize}
thus since $[0..n]$ is well founded, $\elemB^j(\mathit{win})$ gets all values from $[1..|\I_{\mathit{win}}|]$ infinitely often and $v^j\models\I_{\mathit{win}}^k$ infinitely often for each $k$ and the $\mathit{win}$th pair is satisfied. \qed
\end{proof}

\subsection{Completeness}\label{ssec:complete}

We use the standard $\mu$-calculus and define an operator $\pred$ of the \emph{controllable predecessor} as follows: $$\pred\varphi=\{u\in V_0\mid\exists (u,v)\in E: v\models\varphi\}\cup \{u\in V_1\mid\forall (u,v)\in E: v\models\varphi\}$$
Further, we define recursively
\begin{align*}
 \mathfrak{GR}(\emptyset,\varphi,W)&= W\\
 \mathfrak{GR}(\mathcal{GR},\varphi,W)&= \bigvee_{i\in[1..k]}\nu Y.\bigwedge_{j\in[1..|\I_i|]} \mu X.\; \mathfrak{GR}\Big(\mathcal{GR}\setminus\{(F_i,\I_i)\},\varphi\wedge\neg F_i,\\
& ~~~~~~~~~~~~~~~~~~~~~~~~~~~~
 W\vee(\varphi\wedge\neg F_i\wedge I_i^j\wedge \pred Y) \vee (\varphi\wedge\neg F\wedge \pred X) \Big)
\end{align*}
Let $\mathfrak{Win}(\mathcal{GR})$ be the set of winning vertices of Player 0 with the winning condition being the set $\mathcal{GR}$ of GRPs. Then a simple adaptation of Claim 9 of~\cite{PP} to the setting with the conjunction yields an alternative characterization of the winning set.

\begin{lemma}
For non-empty $\mathcal{GR}$, the set $\mathfrak{GR}(\mathcal{GR},\varphi,W)$ is the winning region for Player 0 and the winning condition
$$\bigvee_{(F_,\I)\in\mathcal{GR}} \left( (\varphi\wedge\neg F)\U W\vee \G\Big(\varphi\wedge\neg F\wedge\bigwedge_{j\in[1..|\I|]} \F\I^j\Big) \vee \Big(\mathfrak{Win}(\mathcal{GR}\setminus\{(F,\I)\})\wedge \G(\varphi\wedge\neg F)\Big) \right)$$
\end{lemma}

\noindent As a result, we get directly by \cite{PP} the following symbolic recursive algorithm.

\begin{lemma}
$\mathfrak{Win}(\mathcal{GR})=\mu Z.\;\mathfrak{GR}(\mathcal{GR},\true,\pred Z)$ 
\end{lemma}

\noindent It remains to show a good ranking on the vertices of $\mathfrak{Win}(\mathcal{GR})$ using the characterization above.

\begin{lemma}
There is a good ranking such that for every vertex $v$, from which Player 0 has a winning strategy, there is a permutation $\elemB\in B$ with $r(v,\elemB)\neq\infty$.
\end{lemma}
\begin{proof}
We show how to define good ranking on $\mathfrak{Win}(\mathcal{GR})$. To this end, we use the characterization above written more explicitly in the following algorithm (where all variables are local):

{\sf
\begin{tabbing}
aaa\=aaa\=aaa\=aaa\=aaa\=aaa\=aaa\=\kill\\
Function mainGR(SetOfPairs) \\
\>LeastFix(Z)\\
\>\> Z := Rabin(SetOfPairs, true, $\pred$Z)\\
\>End -- LeastFix(Z)\\
\>Return Z\\
End\\
\\
Function GR(SetOfPairs, Invariant, AlreadyOk)\\
\>Win := 0\\
\>Foreach $(F,\I)\in$ SetOfPairs\\
\>\> RemainingPairs := SetOfPairs$\setminus\{(F,\I)\}$\\
\>\> GreatestFix(Y) \\
\>\>\> Foreach $j\in\{1..|\I|\}$  \\
\>\>\>\> conjunctionX := true\\
\>\>\>\> LeastFix(X) \\
\>\>\>\>\> NewOk := AlreadyOk $\cup$ (Invariant $\cap$ $\neg F$ $\cap$ $I^j$ $\cap$ $\pred$Y) $\cup$ (Invariant $\cap$ $\neg F$ $\cap$ $\pred$X) \\
\>\>\>\>\> If ($|$RemainingPairs$|=0$) \\
\>\>\>\>\>\>\> X := NewOk \\
\>\>\>\>\> Else \\
\>\>\>\>\>\>\>  X := Rabin(RemainingPairs,  Invariant $\cap$ $\neg F$, NewOk) \\
\>\>\>\>\> End -- If ($|$RemainingPairs$|=0$) \\
\>\>\>\> End -- LeastFix(X)\\
\>\>\>\> Let conjunctionX := conjunctionX $\cap$ X \\
\>\>\>End -- Foreach j \\
\>\>\> Let Y := conjunctionX \\
\>\>End -- GreatestFix(Y) \\
\>\>Let Win := Win $\cup$ Y \\
\>End -- Foreach $(F,\I)$\\
\>Return Win \\
End \\
\end{tabbing}
}
%

Similarly to~\cite{PP}, we monitor the call stack of the procedure. We assgin a counter $i$ to each least fixpoint (Z and all nested X's), starting from zero increasing every time next iteration is done. We consider configurations of the program where all greatest fixpoints are in their last iteration.

Each of the states returned by \textsf{mainGR} gets a rank according to the first time it is discovered by the least fixpoints. For a given configuration, let $p_1\cdots p_k$ be the pairs handled by the nested calls of GR and $i_1\cdots i_k$ the $k$ nested values of $i$'s, i.e.  the numbers of iterations of the nested least fixpoints (considering the last calls of the greatest fixpoints) and $i_0$ the value of the counter for Z, and $j_1\ldots j_k$ the nested values of $j$.
Let $X_{p_1\cdots p_k}^{i_0i_1\cdots i_k}$ 
be the current values of the intersection $X$.

Now for every vertex $v$ in the returned set, and $\elemB\in \B$, we set $r(v,\elemB)$ to be the smallest (w.r.t. $>_{lg}$) element $(p_1\cdots p_k,i_0i_1\cdots i_k)$ of the ranking domain where $v\in X_{p_1\cdots p_k}^{i_0i_1\cdots i_k}$ and  $\elemB(\pi_n)=j_n$ for all $n$.

All other pairs $(v,\elemB)$ get rank $\infty$.

The ranking can now be easily shown to be good following~\cite{PP}, since in order to discover a state, its successors must have been already discovered before and have thus a smaller ranking.
\qed
\end{proof}

%% file: app-fig.tex
\section{Figures}\label{ssec:figures} 
We give figures of the game family used for the experimental evaluation in section~\ref{sec:synthesis:experiments}.
 
\centering
\begin{figure}[htbp]
\centering
\begin{minipage}[t]{0.99\linewidth} 
 \tikzset{     
  dim/.style={
           rounded corners,
           minimum height=2em,
           inner sep=2pt,
           text centered,
   },
   p1/.style={
	   draw=black, very thick,
           minimum height=2em,
           inner sep=2pt,
           text centered,
           },
   p0/.style={
           circle,
           draw=black, very thick,
           minimum height=2em,
           inner sep=2pt,
           text centered,
           }
}
\centering
\scalebox{0.75}{
\begin{tikzpicture}[->,>=stealth',scale=1]<2->

 \node[dim, anchor=center] (SINIT) 
 {}; 
 \node[p0, left =of SINIT] (S0) 
 {\hspace{0.5cm}};

 \node[p1, above right =of S0] (S1) 
 {$\qquad$};

 \node[p1, below right =of S0] (S2) 
 {$\qquad$};

 {\path[very thick](S0) edge [loop left] node  {$2^{Ap} \setminus \{\{a\}, \{b\}\}$}  (S0);}

 {\path[very thick](S0) edge [bend left=40]  node [left] {$\{a\}$}  (S1);}
 {\path[very thick](S0) edge [bend right=40] node [left] {$\{b\}$} (S2);}

 {\path[very thick](S1) edge [loop above] node [above] {$\{b\}$}  (S1);}
 {\path[very thick](S2) edge [loop below] node [below]   {$\{c\}$}  (S2);}

 {\path[very thick](S1) edge [bend left=40] node [right] {$\,\,\,2^{Ap} \setminus \{\{b\}\}$}  (S0);}
 {\path[very thick](S2) edge [bend right=40] node [right] {$\,\,\,2^{Ap} \setminus \{\{c\}\}$}  (S0);}

\end{tikzpicture}
}
\end{minipage}
\newline
\begin{minipage}[t]{0.99\linewidth}
\centering
 \tikzset{     
  dim/.style={
           rounded corners,
           minimum height=2em,
           inner sep=2pt,
           text centered,
   },
   p1/.style={
	   draw=black, very thick,
           minimum height=2em,
           inner sep=2pt,
           text centered,
           },
   p0/.style={
           circle,
           draw=black, very thick,
           minimum height=2em,
           inner sep=2pt,
           text centered,
           }
}
\scalebox{0.75}{
\begin{tikzpicture}[->,>=stealth',scale=1]<2->

 \node[dim, anchor=center] (SINIT) 
 {}; 
 \node[p0, left =of SINIT] (S0) 
 {\hspace{0.5cm}};

 \node[p1, above right =of S0] (S1) 
 {$\qquad$};

 \node[p1, below right =of S0] (S2) 
 {$\qquad$};

 \node[p0, right =5cm of S0] (S3) 
 {\hspace{0.5cm}};

 \node[p1, above right =of S3] (S4) 
 {$\qquad$};

 \node[p1, below right =of S3] (S5) 
 {$\qquad$};

 {\path[very thick](S0) edge [loop left] node  {$2^{Ap} \setminus \{\{a\}, \{b\}\}$}  (S0);}

 {\path[very thick](S0) edge [bend left=40]  node [left] {$\{a\}$}  (S1);}
 {\path[very thick](S0) edge [bend right=40] node [left] {$\{b\}$} (S2);}

 {\path[very thick](S1) edge [loop above] node [above] {$\{b\}$}  (S1);}
 {\path[very thick](S2) edge [loop below] node [below]   {$\{c\}$}  (S2);}

 {\path[very thick](S1) edge [bend left=20] node [right] {$\,\,\,2^{Ap} \setminus \{\{b\}\}$}  (S3);}
 {\path[very thick](S2) edge [bend right=20] node [right] {$\,\,\,2^{Ap} \setminus \{\{c\}\}$}  (S3);}

 {\path[very thick](S3) edge [loop left] node  {$2^{Ap} \setminus \{\{a\}, \{b\}\}$}  (S3);}
 {\path[very thick](S3) edge [bend left=40]  node [right] {$\,\{a\}$}  (S4);}
 {\path[very thick](S3) edge [bend right=40] node [right] {$\,\{b\}$} (S5);}

 {\path[very thick](S4) edge [loop above] node [above] {$\{b\}$}  (S4);}
 {\path[very thick](S5) edge [loop below] node [below]   {$\{c\}$}  (S5);}

 {\path[very thick](S4) edge [bend left=40] node [right] {$\,\,\,2^{Ap} \setminus \{\{b\}\}$}  (S3);}
 {\path[very thick](S5) edge [bend right=40] node [right] {$\,\,\,2^{Ap} \setminus \{\{c\}\}$}  (S3);}

\end{tikzpicture}
}
\end{minipage}
\newline
\begin{minipage}[t]{0.99\linewidth}
\centering
 \tikzset{     
  dim/.style={
           rounded corners,
           minimum height=2em,
           inner sep=2pt,
           text centered,
   },
   p1/.style={
	   draw=black, very thick,
           minimum height=2em,
           inner sep=2pt,
           text centered,
           },
   p0/.style={
           circle,
           draw=black, very thick,
           minimum height=2em,
           inner sep=2pt,
           text centered,
           }
}
\scalebox{0.75}{
\begin{tikzpicture}[->,>=stealth',scale=1]<2->

 \node[dim, anchor=center] (SINIT) 
 {}; 
 \node[p0, left =of SINIT] (S0) 
 {\hspace{0.5cm}};

 \node[p1, above right =of S0] (S1) 
 {$\qquad$};

 \node[p1, below right =of S0] (S2) 
 {$\qquad$};

 \node[p0, right =5cm of S0] (S3) 
 {\hspace{0.5cm}};

 \node[p1, above right =of S3] (S4) 
 {$\qquad$};

 \node[p1, below right =of S3] (S5) 
 {$\qquad$};

 \node[p0, right =5cm of S3] (S6) 
 {\hspace{0.5cm}};

 \node[p1, above right =of S6] (S7) 
 {$\qquad$};

 \node[p1, below right =of S6] (S8) 
 {$\qquad$};

 {\path[very thick](S0) edge [loop left] node  {$2^{Ap} \setminus \{\{a\}, \{b\}\}$}  (S0);}
 {\path[very thick](S0) edge [bend left=40]  node [left] {$\{a\}$}  (S1);}
 {\path[very thick](S0) edge [bend right=40] node [left] {$\{b\}$} (S2);}
 {\path[very thick](S1) edge [loop above] node [above] {$\{b\}$}  (S1);}
 {\path[very thick](S2) edge [loop below] node [below]   {$\{c\}$}  (S2);}
 {\path[very thick](S1) edge [bend left=20] node [right] {$\,\,\,2^{Ap} \setminus \{\{b\}\}$}  (S3);}
 {\path[very thick](S2) edge [bend right=20] node [right] {$\,\,\,2^{Ap} \setminus \{\{c\}\}$}  (S3);}

 {\path[very thick](S3) edge [loop left] node  {$2^{Ap} \setminus \{\{a\}, \{b\}\}$}  (S3);}
 {\path[very thick](S3) edge [bend left=40]  node [right] {$\,\{a\}$}  (S4);}
 {\path[very thick](S3) edge [bend right=40] node [right] {$\,\{b\}$} (S5);}
 {\path[very thick](S4) edge [loop above] node [above] {$\{b\}$}  (S4);}
 {\path[very thick](S5) edge [loop below] node [below]   {$\{c\}$}  (S5);}
 {\path[very thick](S4) edge [bend left=40] node [right] {$\,\,\,2^{Ap} \setminus \{\{b\}\}$}  (S6);}
 {\path[very thick](S5) edge [bend right=40] node [right] {$\,\,\,2^{Ap} \setminus \{\{c\}\}$}  (S6);}

 {\path[very thick](S6) edge [loop left] node  {$2^{Ap} \setminus \{\{a\}, \{b\}\}$}  (S6);}
 {\path[very thick](S6) edge [bend left=40]  node [right] {$\,\{a\}$}  (S7);}
 {\path[very thick](S6) edge [bend right=40] node [right] {$\,\{b\}$} (S8);}
 {\path[very thick](S7) edge [loop above] node [above] {$\{b\}$}  (S7);}
 {\path[very thick](S8) edge [loop below] node [below]   {$\{c\}$}  (S8);}
 {\path[very thick](S7) edge [bend left=40] node [right] {$\,\,\,2^{Ap} \setminus \{\{b\}\}$}  (S6);}
 {\path[very thick](S8) edge [bend right=40] node [right] {$\,\,\,2^{Ap} \setminus \{\{c\}\}$}  (S6);}
\end{tikzpicture}
}
\end{minipage}
\caption[Games for Experiments]{Games used for the experiments in section~\ref{sec:synthesis:experiments},
with $Ap = \{a,b,c\}$.}
\label{fig:games:experiments}
\end{figure}
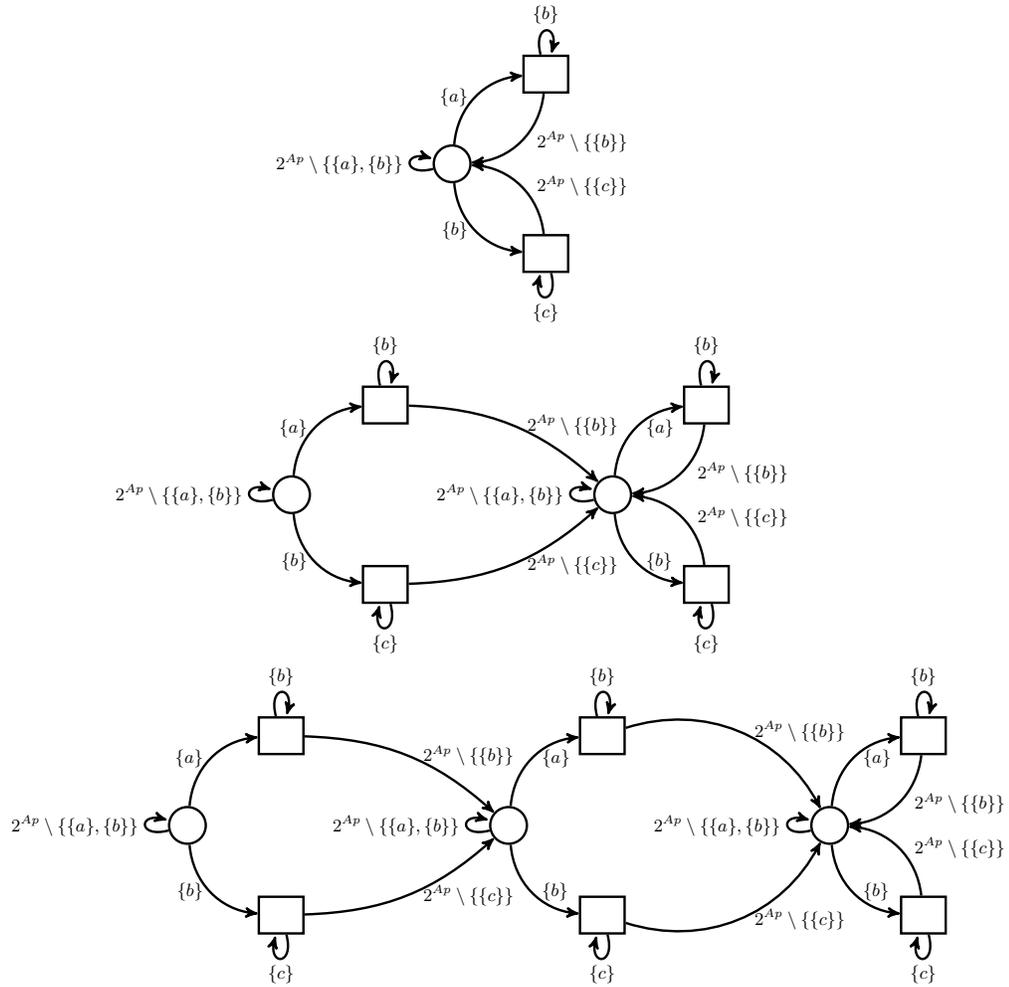